\documentclass[a4paper,aps,pra,10pt,twocolumn,floatfix,superscriptaddress,notitlepage,usenames,dvipsnames,svgnames,table,nofootinbib,longbibliography]{revtex4-2}
\usepackage{float}       
\usepackage{subcaption}  
\usepackage{adjustbox}   

\usepackage{tikz}
\usetikzlibrary{quantikz2}

\usepackage[margin=1in]{geometry}

\usepackage{graphicx} 
\usepackage{subcaption}
\usepackage{adjustbox}

\usepackage{amsmath}
\usepackage{braket}
\DeclareMathOperator{\Tr}{Tr}

\usepackage{ulem}

\usepackage{amssymb}
\usepackage{amsthm}
\usepackage{hyperref}
\usepackage{cleveref}

\newtheorem{theorem}{Theorem}[section]
\newtheorem{corollary}{Corollary}[theorem]
\newtheorem{lemma}[theorem]{Lemma}
\newtheorem{definition}{Definition}
\newtheorem{proposition}{Proposition}

\usepackage{xcolor}

\usepackage{pgfplots}
\definecolor{blue-violet}{rgb}{0.54, 0.17, 0.89}

\usepackage{hyperref}
\hypersetup{
 pdfauthor={Emanuel Dallas, Faidon Andreadakis, and Paolo Zanardi},
 pdftitle={A Butterfly Effect in Encoding-Decoding Quantum Circuits},
 pdflang={English},colorlinks=true,linkcolor=RubineRed,citecolor=blue-violet,urlcolor=Cerulean}

\begin{document}

\title{A Butterfly Effect in Encoding-Decoding Quantum Circuits}

\author{Emanuel Dallas}
\email [e-mail: ]{dallas@usc.edu}
\affiliation{Department of Physics and Astronomy, and Center for Quantum Information Science and Technology, University of Southern California, Los Angeles, California 90089-0484, USA}
\author{Faidon Andreadakis}
\email [e-mail: ]{fandread@usc.edu}
\affiliation{Department of Physics and Astronomy, and Center for Quantum Information Science and Technology, University of Southern California, Los Angeles, California 90089-0484, USA}
\author{Paolo Zanardi}
\email [e-mail: ]{zanardi@usc.edu}
\affiliation{Department of Physics and Astronomy, and Center for Quantum Information Science and Technology, University of Southern California, Los Angeles, California 90089-0484, USA}
\affiliation{Department of Mathematics, University of Southern California, Los Angeles, California 90089-2532, USA}

\date{\today}
\begin{abstract}
The study of information scrambling has profoundly deepened our understanding of many-body quantum systems. Much recent research has been devoted to understanding the interplay between scrambling and decoherence in open systems. Continuing in this vein, we investigate scrambling in a noisy encoding-decoding circuit model. Specifically, we consider an $L$-qubit circuit consisting of a Haar-random unitary, followed by noise acting on a subset of qubits, and then by the inverse unitary. Scrambling is measured using the bipartite algebraic out-of-time-order correlator ($\mathcal{A}$-OTOC), which allows us to track information spread between extensively sized subsystems. We derive an analytic expression for the $\mathcal{A}$-OTOC that depends on system size and noise strength. In the thermodynamic limit, this system displays a \textit{butterfly effect} in which infinitesimal noise induces macroscopic information scrambling. We also perform numerical simulations while relaxing the condition of Haar-randomness, which preliminarily suggest that this effect may manifest in a larger set of circuits.
\end{abstract}

\maketitle

\section{Introduction}\label{intro}

The study of information spreading in quantum systems from local to non-local degrees of freedom, termed ``information scrambling," has profoundly deepened the understanding of many-body quantum systems in recent years. It has elucidated phenomena ranging from black holes \cite{blackhole1,blackhole2,blackhole3} to thermalization in closed quantum systems \cite{therm1} while also bolstering the underpinnings of the established field of quantum chaos \cite{NamitBipartite,chaos1,otoc2}. A wide range of diagnostic tools have been employed to probe quantum chaos, such as Hamiltonian spectral statistics \cite{spectral1,spectral2,spectral3}, the Loschmidt echo \cite{Loschmidt_1,Loschmidt_2,Loschmidt_3}, and out-of-time-order correlators (OTOCs) \cite{otoc1,otoc2,otoc3}. Strongly chaotic systems may even display exponential sensitivity to small perturbations, a phenomenon familiarly known as the \textit{butterfly effect} \cite{butterfly1,butterfly2}.

Open system dynamics introduce an additional layer of complexity to these phenomena. Now, there is a ``competition" between scrambling dynamics and information lost to decoherent noise \cite{aotoc_faidon}. Measurement-induced phase transitions comprise a class of such phenomena in which this competition generates complex behavior \cite{mipt1,mipt2,mipt3}. Recently, Lovas \textit{et al.} have demonstrated that boundary dissipation induces a ``quantum coding transition" between preserved and fully lost information in a Haar-random brickwork circuit \cite{Lovas}. Turkeshi and Sierant have analyzed an ``encoding-decoding" circuit model in which ancillary qubits are appended to an initial multi-qubit state; all qubits are then scrambled by a random unitary, acted on by local noise, and unscrambled; then the ancillas are projected onto their initial state \cite{Turkeshi}. They observed an ``error-resilience phase transition" of the fidelity between the initial and final state.

This paper, inspired by these recent results and models, investigates a similar encoding-decoding circuit model which integrates scrambling, chaos, and open system dynamics, yielding analytically tractable results. Specifically, we consider an $L$-qubit circuit consisting of a Haar-random unitary $U$, followed by identical noise processes on $k$ qubits, followed by $U^\dagger$. This simple model allows us to clearly observe the effects of increasing noise strength in a state-independent manner. Of course, local Hamiltonians generate a much more restricted set of unitary dynamics than the full unitary group. However, random matrices provide an analytically tractable proxy for ``maximally scrambling'' dynamics, and the scrambling properties of locally interacting chaotic systems have been observed to quickly equilibrate to near the typical value predicted by random matrix theory \cite{rmt1,rmt2,rmt3}.

Our metric of scrambling is the bipartite form of the algebraic out-of-time-order correlator ($\mathcal{A}$-OTOC \cite{aotoc_faidon}), introduced in \cite{StyliarisBipartite} for closed and \cite{NamitBipartite} for open systems. The bipartite $\mathcal{A}$-OTOC tracks the scrambling between two subsystems and has already been demonstrated as a tool for distinguishing between integrable and chaotic dynamics \cite{StyliarisBipartite,NamitBipartite}. The $A$-OTOC, involving a Haar average over observables in each subsystem, has the nice property of depending only on the selected bipartition and dynamics, allowing us to glean very general insights. In this paper, we select for the two subsystems each half of the qubit chain to probe scrambling that is extensive in system size.

We derive an analytic formula for the $A$-OTOC that depends on the noise channel $\mathcal{E}$ and $L$. In the limit that $L \rightarrow \infty$, the formula dramatically simplifies. From this, we immediately observe a butterfly effect in which noise on finitely many qubits causes ``macroscopic" scrambling. We also derive expressions for the $A$-OTOC under extensively scaling noise ($k \propto L)$ and connect these to the competition between decoherence and scrambling. We then present multiple explicit examples and numerically demonstrate that noisy brickwork circuits accord well with theory, allowing for the possibility of observation on real quantum hardware and providence evidence that these phenomena may extend beyond Haar-random unitaries.

In \cref{mathprelim}, we provide mathematical background on the $A$-OTOC. In \cref{physsetup}, we introduce the physical model we will analyze. In \cref{mainresults}, we present formulae for the $A$-OTOC under various limits and noise channels, and derive a bound on it in the $L \rightarrow \infty$ limit. In \cref{examples}, we present two examples that analytically demonstrate the butterfly effect; we also numerically simulate the two examples on finitely many qubits and compare with the theory. Finally, \cref{conclusion} provides a summary and avenues for future work.

\section{Mathematical preliminaries}\label{mathprelim}

\subsection{The general \texorpdfstring{$\mathcal{A}$-OTOC}{A-OTOC}}

Let $\mathcal{H} \cong \mathbb{C}^d$ be a finite-dimensional Hilbert space representing a quantum system. Any physical observable within this system is described by a linear operator, and we denote the space of all such operators by $\mathcal{L}(\mathcal{H})$. This space, $\mathcal{L}(\mathcal{H})$, also forms a Hilbert space, equipped with the Hilbert-Schmidt inner product, defined as $\langle X, Y\rangle=\Tr\left(X^{\dagger} Y\right)$ for operators $X$ and $Y$. We denote $\lVert X \rVert_2^2 = \Tr(X^{\dagger}X)$. In the Heisenberg picture, the evolution of a physical observable $X \in \mathcal{L}(\mathcal{H})$ in an open quantum system is given by $\mathcal{E}(X)$, where $\mathcal{E}^\dagger$ is the CPTP map governing state evolution.

Our quantities of interest make use of averages over the Haar-random unitaries. These are operators sampled over the Haar measure, the unique probability measure that is left- and right-invariant over the full unitary group; these properties ensure the measure is ``uniform" over the unitary group (see \cite{Mele2024introductiontohaar} for a more extensive overview).

The key mathematical objects in the present work are Hermitian-closed, unital subalgebras $\mathcal{A} \subset \mathcal{L}(\mathcal{H})$, which describe the degrees of freedom of interest. The commutant algebra $\mathcal{A}^{\prime}=\left\{Y \in \mathcal{A}^{\prime} \mid[Y, X]=0 \text{ for all } X \in \mathcal{A}\right\}$ captures the symmetries of $\mathcal{A}$ and corresponds to degrees of freedom that are initially uncorrelated with $\mathcal{A}$. By the double commutant theorem \cite{doublecommutantthm}, $\left(\mathcal{A}^{\prime}\right)^{\prime}=\mathcal{A}$, so these algebras can be viewed as pairs $\left(\mathcal{A}, \mathcal{A}^{\prime}\right)$. During time evolution, information is exchanged between $\mathcal{A}$ and $\mathcal{A}^{\prime}$, which we quantify with the $\mathcal{A}$-OTOC \cite{aotoc_faidon}.

\begin{definition}\label{aotoc_def}
    The \textbf{$\mathcal{A}$-OTOC} of algebra $\mathcal{A}$ and channel $\mathcal{E}$ is defined as:
    \begin{equation}
        G_{\mathcal{A}}\left(\mathcal{E}\right)=\frac{1}{2 d} \mathbb{E}_{X_{\mathcal{A}}, Y_{\mathcal{A}^{\prime}}}\left[\left\|\left[X_{\mathcal{A}}, \mathcal{E}\left(Y_{\mathcal{A}^{\prime}}\right)\right]\right\|_2^2\right]
    \end{equation}
$\mathbb{E}_{X_{\mathcal{A}}, Y_{\mathcal{A}^{\prime}}}$ denotes the Haar average over the unitaries $X_{\mathcal{A}} \in \mathcal{A}$ and $Y_{\mathcal{A}^{\prime}} \in \mathcal{A}^{\prime}$. 
\end{definition}

The evolution of operators in $\mathcal{A}^{\prime}$ under $\mathcal{E}$ may lead to non-commutativity with operators in $\mathcal{A}$, which is interpreted as scrambling between the corresponding degrees of freedom.

In this paper, we reserve $\mathcal{E}$ to refer to the map in the Heisenberg picture, and we consider maps that are CPTP in the Schrodinger picture (i.e. $\mathcal{E}^\dagger$ is CPTP). This implies that $\mathcal{E}$ is unital.

\subsection{Bipartite \texorpdfstring{$\mathcal{A}$-OTOC}{A-OTOC}}

We now consider the case of bipartite algebras. Let the Hilbert space $\mathcal{H}_{AB} \equiv \mathcal{H}_A \otimes \mathcal{H}_B$. Our algebra is now $\mathcal{A} = \mathbb{I}_A\otimes L(\mathcal{H}_B)$, with $\mathcal{A}' = L(\mathcal{H}_A)\otimes \mathbb{I}_B$. Henceforth, we will denote the bipartite $\mathcal{A}$-OTOC as $G$.

In \cite{NamitBipartite}, the following was shown: 
\begin{proposition}\label{aotoc_swap}
    Let $S \equiv S_{A A^{\prime} B B^{\prime}}$ be the swap operator over $\mathcal{H}_{A B} \otimes \mathcal{H}_{A^{\prime} B^{\prime}} \ (= \mathcal{H}_A \otimes \mathcal{H}_{B} \otimes \mathcal{H}_{A'} \otimes \mathcal{H}_{B'} )$ which swaps operators on $\mathcal{H}_A$ with $\mathcal{H}_{A'}$ and $\mathcal{H}_B$ with $\mathcal{H}_{B'}$. Similarly, $S_{AA'}$ swaps operators on $\mathcal{H}_A$ with $\mathcal{H}_{A'}$. Then for a quantum channel $\mathcal{E}: \mathcal{L}\left(\mathcal{H}_{A B}\right) \rightarrow$ $\mathcal{L}\left(\mathcal{H}_{A B}\right)$, the bipartite $\mathcal{A}$-OTOC is:
    \begin{equation}
        G(\mathcal{E})=\frac{1}{d^2} \operatorname{Tr}\left(\left(d_B S-S_{A A^{\prime}}\right) \mathcal{E}^{\otimes 2}\left(S_{A A^{\prime}}\right)\right).
    \end{equation}
\end{proposition}

This is the formula we use to derive our results in \cref{mainresults}.

An alternative formula for the bipartite $\mathcal{A}$-OTOC was also derived \cite{NamitBipartite}. Let us denote $\tilde{\mathcal{E}}(X) \equiv \mathcal{E}(X \otimes \frac{\mathbb{I}}{d_B})$. Then:
\begin{proposition}\label{aotoc_state}
    We denote by $\psi:=|\psi\rangle\langle\psi|$ with $|\psi\rangle \in \mathcal{H}_A$. Then,
    \begin{equation}
        G(\mathcal{E})=N_A \mathbb{E}_\psi\left[S_L\left(\operatorname{Tr}_B \widetilde{\mathcal{E}}(\psi)\right)-d_B\left(S_L(\widetilde{\mathcal{E}}(\psi))-S_L^{\min }\right)\right]
    \end{equation}
where $\mathbb{E}_\psi$ is the the Haar average over $\mathcal{H}_A, N_A:=\frac{d_A+1}{d_A}$, the linear entropy $S_L(\rho) \equiv 1 - \Tr(\rho^2)$, and $S_L^{\min }:=1-\frac{1}{d_B}$. 
\end{proposition}

We make use of this formula for numerics (\cref{numerics}), in which it is easier to sample a finite set of states to reach an approximate value for $G(\mathcal{E})$. For the remainder of this paper, we will take $\mathcal{H}_A\cong (\mathbb{C}^2)^{\otimes L/2}$ to be the first set of $L/2$ qubits and $\mathcal{H}_B\cong (\mathbb{C}^2)^{\otimes L/2}$ to be the second set of $L/2$ qubits.

\section{Physical Setup}\label{physsetup}

We consider a chain of $L$ qubits, where $L$ is even. 
Let $U$ be an $L$-qubit unitary, 
and let $\mathcal{E}$ be a single-qubit CP map. Let the adjoint action of $U$, in the Heisenberg picture, be denoted as $\mathcal{U}$, so that an operator $A$ gets sent to $\mathcal{U}(A)$ by the action of $U$.

We consider the quantum map $\mathcal{C} \equiv \mathcal{U^\dagger}\mathcal{E}^{\otimes k}\mathcal{U}$ in the Heisenberg picture. For what follows, without loss of generality, we can assume that $\mathcal{E}$ is applied to the ``first" $k$ qubit sites.

\begin{figure}
    \centering
    \begin{quantikz}[row sep = {0.6cm, between origins}]
    &\gate[4, style={rounded corners, fill=red!20}][1cm]{U}&[0.3cm]\measure[style = {fill=yellow}]{\mathcal{E}}\gategroup[wires=2,steps=2,style={inner xsep=0pt, inner ysep=0pt, outer sep=0pt, draw=none}, label style={label position=right, anchor=east, xshift=-0.75cm, yshift = -0.2cm}]{\parbox{3cm}{\baselineskip=0.4\baselineskip k qubits \\ \makebox[1.3cm]{.} \\ \makebox[1.3cm]{.} \\\makebox[1.3cm]{.}}}
    &[0.5cm]\gate[4, style={rounded corners, fill=red!20}][1cm]{U^\dagger}&\\[1.2cm]
    &&\measure[style = {fill=yellow}]{\mathcal{E}}&& \\
    &&&& \\
    &&&& 
    \end{quantikz}

    \caption{Circuit diagram of $\mathcal{C}$.}
\end{figure}

The map $\mathcal{C}$ corresponds to a unitary evolution, a (possibly) noisy perturbation, and an undoing of the unitary evolution. Analyzing the information scrambling properties of this map gives insight into the extent to which local, possibly infinitesimal perturbations can cause global, macroscopic effects -- namely, this is an example of quantum chaos.
To make analytic results tractable, we calculate the average bipartite $\mathcal{A}$-OTOC over Haar-random $U$, denoted as $\overline{G(\mathcal{C})} \equiv \mathbb{E}_U G(\mathcal{U}^\dagger \mathcal{E}^{\otimes k} \mathcal{U})$. For large $L$, the phenomenon of measure concentration implies that typical Haar-random unitaries fall close to this average \cite{measureconcentration}, so it gives insight into generic behavior. We expand on this with a formal proof in \cref{app_typicality}.

\subsection{Connection to Loschmidt Echo}

The bipartite $\mathcal{A}$-OTOC in this model can be viewed qualitatively as a sort of generalization of the Loschmidt echo, which we now define.

Consider a state $\ket{\psi}$ evolving under a Hamiltonian $H$ with small perturbation $V$ for time $t$. The Loschmidt echo, introduced in \cite{Loschmidt_1}, is defined as:
\begin{equation}\label{loschmidt_echo}
    M(t) \equiv \left| \bra{\psi}e^{i(H+V)t}e^{-iHt}\ket{\psi}\right|^2.
\end{equation}

Namely, the Loschmidt echo measures the fidelity between a state evolved by the perturbed Hamiltonian for time $t$ and the same state evolved by the unperturbed Hamiltonian for the same time. Another, equivalent intuition is that it tracks how far $\ket{\psi}$ deviates from itself after evolving forward in time under $H$ and backward in time under $H+V$. As such, the behavior of $M(t)$ with respect to $t$ can reveal to what extent small perturbations yield macroscopic changes, capturing the chaoticity of $H$. Its connections to scrambling and chaos have been studied extensively \cite{Loschmidt_2,Loschmidt_3}.

$G(\mathcal{C})$ is an average over operators, capturing more general features of the dynamics $U$ than a state- or operator-dependent quantity. However, like the Loschmidt echo, it captures a similar notion of evolving ``forward" by $U$ and $\mathcal{E}^{\otimes k}$ and ``backward" under $U^\dagger$, specifically tracking the propensity of the noise to induce scrambling. Understanding the behavior of $G(\mathcal{C})$ gives insight into the chaoticity of a given $U$.

\section{Main results}\label{mainresults}

To arrive at results for this model, we insert $\mathcal{C}$ in place of the generic $\mathcal{E}$ in \cref{aotoc_swap}. This expression contains four copies of the unitary channel $\mathcal{U}$. This means that we only need unitaries drawn from a 4-design \cite{t_design} to reproduce the Haar average. We can analytically compute this average by linearizing the expression and using Schur-Weyl duality \cite{weingarten_1}. The full details are presented in \cref{appendix}. Since this average involves a large number of terms and group-theoretic data as inputs, we created a Mathematica notebook that performs the analytic calculations.

To simplify our results, we invoke the following:

\begin{definition}[Natural representation]\label{naturalrep}
    Let $\mathcal{E}$ have the Kraus representation:
\begin{equation}
    \mathcal{E} = \sum_i K_i (\cdot) K_i^\dagger.
\end{equation}

Then $X \equiv \sum_i K_i \otimes K_i^\dagger$ is the \textbf{natural representation} of $\mathcal{E}$.
\end{definition}

One can show that this representation is unique for a given map, i.e. it is independent of any particular choice of Kraus operators \cite{naturalrepresentation}. Since we work only with generic $\mathcal{E}$, we omit any subscript label on $X$ that explicitly refers to $\mathcal{E}$.

This all results in the following:
\begin{lemma}[$\overline{G(\mathcal{C})}$ for finite qubit chain]
Let $\mathcal{E}$ be a CPTP map. Let $\rho_{max}$ be the maximally mixed single-qubit state, $\gamma$ be the purity, and $S$ be the swap operator on $\mathcal{H}^{\otimes 2} \equiv \mathcal{H}_1\otimes\mathcal{H}_2$. Then 
\begin{widetext}
\begin{equation}
    \begin{split}
        \overline{G(\mathcal{C})} =& 
        \left (\frac{-2^{-4k}\cdot \left(2^{-L} - 2^{-2L} \right)}{\left(-3+2^L\right)\left(-2+2^L\right)\left(1+2^L\right)\left(2+2^L\right)\left(3+2^L\right)}\right)\cdot
        \\ 
        &\Bigg[ \left(2^{6L}-2^{4L+2}\right) \Tr(X)^{2k} - 6\cdot 2^k\left(2^{4L}-2^{2L+2}\right)\Re\{(\Tr(\Tr_1{X})^2)^k)\}
        \\
        & - 2^{2k+1}\left(2^{4L}-2^{2L+1}\right)\Tr(X)^k - 2^{k+1}\left(2^{4L}+3\cdot 2^{2L+1}\right)(\Tr(\Tr_1 X\Tr_2 X))^k
        \\
        &  - 5\cdot 2^{2k+2} 2^{2L}\Re\{(\Tr(\mathcal{E}\otimes \mathbb{I}(X)))^k\} - 2^{3k+1}\left(2^{2L} + 6\right)\Tr(SX^2)^k + 2^{2k}\left(2^{4L}-2^{2L+2} \right)(\Tr(X^2))^k 
        \\
        & - 2^{2k}\left(2^{6L} - 3\cdot 2^{4L+2} +22\cdot 2^{2L} \right)\left\lvert\left\lvert X\right\rvert\right\rvert_2^{2k} + 2^{5k}\left(2^{4L} - 3\cdot 2^{2L+2} +12 \right)\gamma(\mathcal{E}^\dagger(\rho_{max}))^k \Bigg].
    \end{split}
\end{equation}
\end{widetext}
\end{lemma}

From this, we consider the thermodynamic limit in which $L \rightarrow \infty$. The first parenthetical term goes to $-2^{-4k} \cdot 2^{-6L}$. Thus, the only surviving terms are ones that, within the square brackets, have a prefactor of order at least $2^{6L}$. We can easily see that only two terms do: $\left(2^{6L}-2^{4L+2}\right) \Tr(X)^{2k}$ and $- 2^{2k}\left(2^{6L} - 3\cdot 2^{4L+2} +22\cdot 2^{2L} \right)\left\lvert\left\lvert X\right\rvert\right\rvert_2^{2k}$.

Note that all terms that do not survive in the thermodynamic limit decay exponentially, in powers of $2^L$. This means the finite-$L$ expression approaches the below $L \rightarrow \infty$ exponentially quickly, explaining the near indistinguishability of the two in the plots in \cref{numerics}. Keeping only the highest-ordered terms, we arrive at the following result.

\begin{theorem}[$\overline{G(\mathcal{C})}$ for infinite qubit chains]\label{maintheorem} Let $\mathcal{E}$ be a CPTP map. In the limit that $L \rightarrow \infty$, the Haar-averaged symmetric bipartite $A$-OTOC is given by:
\begin{equation}
    \overline{G(\mathcal{C})} = \left|\left|\frac{X}{2}\right|\right|_2^{2k} - \left(\frac{\Tr(X)}{4} \right)^{2k}.
\end{equation}
    
\end{theorem}

This is a rather striking result. In general, for an arbitrary $\mathcal{E}$, $\overline{G(\mathcal{C})}$ is finite, even when $k$ is finite, i.e. $\frac{k}{L} \rightarrow 0$. An infinitesimal perturbation thereby propagates finite scrambling across an extensively scaling bipartition. One may intuitively expect such a perturbation to be vanishing in the $L \rightarrow \infty$ limit, but this is an instance of a strong butterfly effect.

We can gain intuition about this result by examining the Kraus operators of $\mathcal{C}$. Let $\mathcal{E}^{\otimes k}$ have the Kraus operators $\{K_i^\dagger \}_i$ in the Heisenberg picture. Since $\mathcal{E}^{\otimes k}$ is a $k$-qubit map, any Kraus operator $K_j$ of $\mathcal{E}^{\otimes k}$ has support on at most $k$ qubits. Now consider the evolution of an operator $Q$ under the circuit $\mathcal{C}$. $Q$ evolves to $\sum_i UK_i^\dagger U^\dagger Q U^\dagger K_i U$. We see that $\mathcal{C}$ can be expressed as a CPTP map whose Kraus operators are $\{U^\dagger K_i^\dagger U \}_i$;  in other words, they are the Kraus operators of $\mathcal{E}^{\otimes k}$ conjugated by $U$. Since $U$ is a highly non-local, random unitary, it generally maps operators, even initially $k$-local ones, to other highly non-local ones. The Kraus operators of the full circuit $\mathcal{C}$ tend to have full support on all $L$ qubits. Thus, regardless of $k$, so long as it is greater than zero, $\mathcal{C}$ maps the initially $L/2$-local operators of $\mathcal{A'} \equiv L(\mathcal{H}_A)\otimes \mathbb{I}_B$ (see \cref{aotoc_swap}) to ones with support on all $L$ qubits.


\begin{corollary}[Maximum value for extensively scaling noise]\label{max_corollary}
Let $k = \alpha L$, where $\alpha \in (0, 1]$, and let $L \rightarrow \infty$.
Then the maximum value of $  \overline{G(\mathcal{C})} = 1$ is attained if and only if $\mathcal{E} = \mathcal{V}$, where $V$ is the adjoint action of any \textit{non-identity} unitary $V$.
    
\end{corollary}

\begin{proof}
     By the triangle inequality,
    \begin{equation}
        \begin{split}
            \left|\left|\frac{X}{2}\right|\right|_2^{2k} &= \left|\left|\frac{1}{2}\sum_i K_i \otimes K_i^\dagger \right|\right|_2^{2k} \leq \sum_i \left|\left|\frac{1}{2} K_i \otimes K_i^\dagger \right|\right|_2^{2k}
            \\
            & = \left(\frac{1}{2}\sum_i\left|\left| K_i\right|\right|^2\right)^{2k} = 
            \left(\frac{1}{2}\sum_i \Tr(K_i^\dagger K_i)\right)^{2k} \\
            &= \left(\frac{1}{2}\Tr(\mathbb{I})\right)^{2k} = 1,
        \end{split}
    \end{equation}
    where the second-to-last equality follows from the completeness of the Kraus operators. 

    We know that the triangle inequality is saturated if and only if all Kraus operators are linearly dependent. Let us assume they are all proportional to some operator $K$, such that $K_i \equiv \alpha_i K$. In this case, $\sum_i K_i^\dagger K_i = c K^\dagger K = \mathbb{I}$, where $c = \sum_i \left|\alpha_i\right|^2$. Then $\sqrt{c}K \equiv V$ satisfies $V^\dagger V = \mathbb{I}$ and is therefore unitary. So for non-unitary maps, $\left(\left|\left|\frac{X}{2} \right|\right|\right)^{2k}$ is less than $1$ and goes to $0$ as $k \rightarrow \infty$. For a unitary map, $\left(\left|\left|\frac{X}{2} \right|\right|\right)^{2k}$ simply equals $\left(\frac{1}{2}\Tr(U^\dagger U)\right)^{2k} = 1$, saturating the inequality.

    We can rewrite $\left(\frac{\Tr(X)}{4}\right)^{2k}$ as $\left(\frac{1}{4}\right)^{2k}\cdot\left(\sum_i \left|\langle K_i, \mathbb{I} \rangle \right|^2\right)^{2k}$. By the Cauchy-Schwarz inequality, we see that
    \begin{equation}
        \left(\sum_i \left|\langle K_i, \mathbb{I} \rangle \right|^2\right)^{2k} \leq \left(\sum_i \left|\left| K_i \right|\right|^2 d\right)^{2k} = \left(d^2\right)^{2k} = 4^{2k},
    \end{equation}
    where the penultimate equality follows again from completeness of the Kraus operators. Since the Cauchy-Schwarz inequality is saturated if and only if the two operators are linearly dependent, \textit{all} non-identity maps have $\frac{\Tr(X)}{4} < \frac{4}{4} = 1$, so the second term goes to $0$ as $k \rightarrow \infty$.

     Thus, any non-trivial unitary satisfies  $\overline{G(\mathcal{C})} = 1$ as $k \rightarrow \infty$, since the first term equals $1$ and the second term goes to $0$ \footnote{The identity map clearly gives $1$ for both terms, so that the $\mathcal{A}$-OTOC is always $0$.}. Any non-unitary map has a first term that goes to $0$ and a second term that goes to $0$ as $k = \alpha L \rightarrow \infty$.
     
\end{proof}

Via this proof, we have thus also proved the following:

\begin{corollary}[$\mathcal{A}$-OTOC for extensively scaling non-unitary noise]\label{corollary_dissipation}
Let $k = \alpha L$, where $\alpha \in \{0, 1\}$, and let $L \rightarrow \infty$. Let $\mathcal{E}$ be a non-unitary CPTP map. Then $\overline{G(\mathcal{C})} = 0$.
\end{corollary}

We emphasize that, for open systems, $G(\mathcal{C}) = 0$ does not imply that no scrambling has occurred. Of course, if $\mathcal{C}$ contains no noise, then the circuit in fact acts as the identity and does not scramble any information. However, $G(\mathcal{C})$ is a measure of both scrambling \textit{and} decoherence. The other scenario in which $G(\mathcal{C}) = 0$ is when there is sufficient decoherent noise such that all of the information scrambled from $A$ to $B$ has dissipated out of the full system. This is the underlying reason for the result in \cref{corollary_dissipation}. 

\section{Examples}\label{examples}

\subsection{Unitary noise (Bloch sphere rotation)}

Any single-qubit unitary can be represented by a rotation of the Bloch sphere by an angle $\theta$ about some axis $\hat{n}$. Let $V$ be such a unitary, and let $\mathcal{E} \equiv \mathcal{V}$ act on $k$ qubits. Then
\begin{equation}
    \overline{G(\mathcal{C})} = 1 - \cos^{4k}(\theta). \label{eq:unitary_k}
\end{equation}
\begin{proof}
    We can represent $V \equiv e^{i\theta\hat{n}\cdot\vec{\sigma}}$. Then $X = e^{i\theta\hat{n}\cdot\vec{\sigma}} \otimes e^{-i\theta\hat{n}\cdot\vec{\sigma}}$. Since we have only one Kraus operator here, the above triangle inequality is an equality, and $\left|\left|\frac{X}{2}\right|\right|_2 = \frac{1}{2}\left|\left| V\right|\right|^2 = \frac{1}{2}\Tr(V^\dagger V) = \frac{1}{2}\Tr(\mathbb{I}) = 1$.

    Now, for the second term, $\Tr(X) = \Tr(V)\cdot\Tr(V^\dagger).$ We know that $e^{\pm i\theta\hat{n}\cdot\vec{\sigma}} = \cos(\theta)\mathbb{I} \pm  i\sin(\theta)\hat{n}\cdot\vec{\sigma}$. Since $\hat{n}\cdot\vec{\sigma}$ is traceless, then $\Tr(X) = \cos^2(\theta)\Tr(\mathbb{I})^2 = 4\cos^2(\theta)$. So $\left(\frac{\Tr(X)}{4} \right)^{2k} = \cos^{4k}(\theta)$, and we get the above equation.
\end{proof}

\begin{figure}
    \centering
    \includegraphics[width=0.6\linewidth]{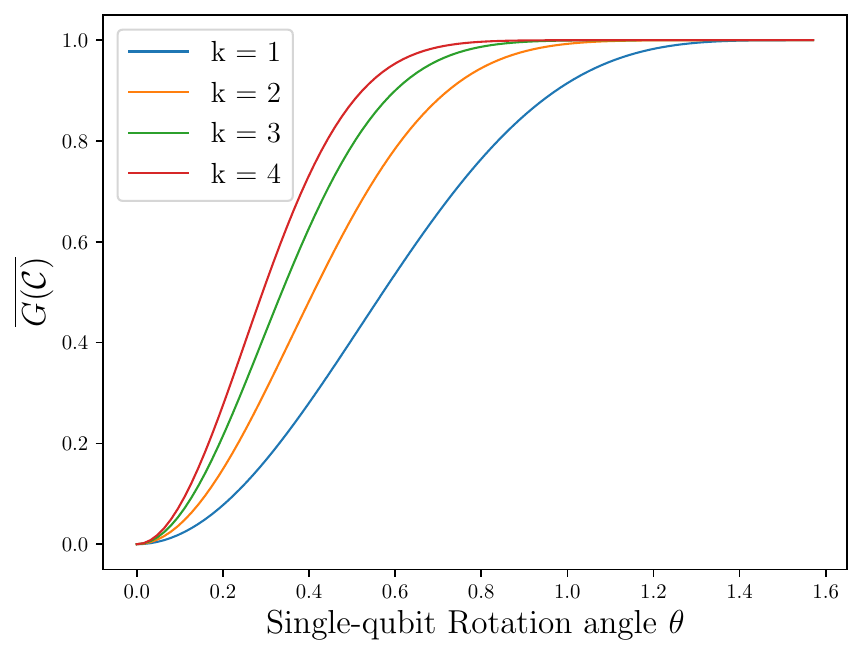}
    \caption{$\overline{G(\mathcal{C})}$ vs. $\theta$ for small values of $k$. As $k$ increases, $\overline{G(\mathcal{C})}$ approaches the step function $H(\theta)$.}
    \label{unitary_k_plot}
\end{figure}

The function in \cref{eq:unitary_k} approaches $1$ for increasing $\theta$ more quickly as $k$ increases. We clearly see that the unitary corollary holds. If $k = \alpha L \rightarrow \infty$ and $V$ is non-trivial, then the second term goes to $0$ and $\overline{G(\mathcal{C})} \rightarrow 1$. Formally, for any $k \propto L$, there is a first-order transition, in the thermodynamic limit, from a non-scrambling phase to a maximally scrambled point at $\theta = 0$. 

Notice that for any (non-zero) value of $k$, we can achieve any value of $\overline{G(\mathcal{C})}$ by varying $\theta$. This is a rather surprising result that warrants emphasis. By applying a $\pi/2$ rotation to a \textit{single} qubit, we can achieve \textit{maximal} scrambling between the two \textit{infinitely large} halves of our bipartition, displaying a quintessential example of the butterfly effect.

\subsection{Depolarization}

Consider a channel $\mathcal{E}$ that fully depolarizes with probability $p$ and acts as the identity with probability $1-p$, such that, for an operator $O$, $\mathcal{E}(O) = p\mathbb{I} + (1-p)O$. Then $\mathcal{E}$ has the Kraus decomposition $\left\{\sqrt{1-\frac{3p}{4}}\mathbb{I}, \frac{\sqrt{p}}{2}\sigma^x,\frac{\sqrt{p}}{2}\sigma^y, \frac{\sqrt{p}}{2}\sigma^z \right\}$. Then \begin{equation}
    \overline{G(\mathcal{C})} = \left( 1 - \frac{3p}{2} + \frac{3p^2}{4}\right)^k - \left(1-\frac{3p}{4}\right)^{2k}.
\end{equation}

\begin{proof}
    It follows that
    \begin{equation}
        X = \left(1-\frac{3p}{4}\right)\mathbb{I}\otimes\mathbb{I} + \frac{p}{4}\vec{\sigma}\cdot\vec{\sigma}.
    \end{equation}
    To calculate $\left|\left|X\right|\right|^2$, we use the basic properties that the Paulis are traceless and square to identity. This gives:
    \begin{equation}
        \left|\left|X\right|\right|^2 = 4\left( \left(1-\frac{3p}{4}\right)^2 + 3\left(\frac{p}{4}\right)^2 \right) = 4 - 6p + 3p^2.
    \end{equation}

    And we can directly calculate:
    \begin{equation}
        \frac{\Tr(X)}{4} = \frac{1}{4}\cdot 4 \cdot \left(1-\frac{3p}{4}\right) = \left(1-\frac{3p}{4}\right).
    \end{equation}

\end{proof}

\begin{figure}
    \centering
    \includegraphics[width=0.9\linewidth]{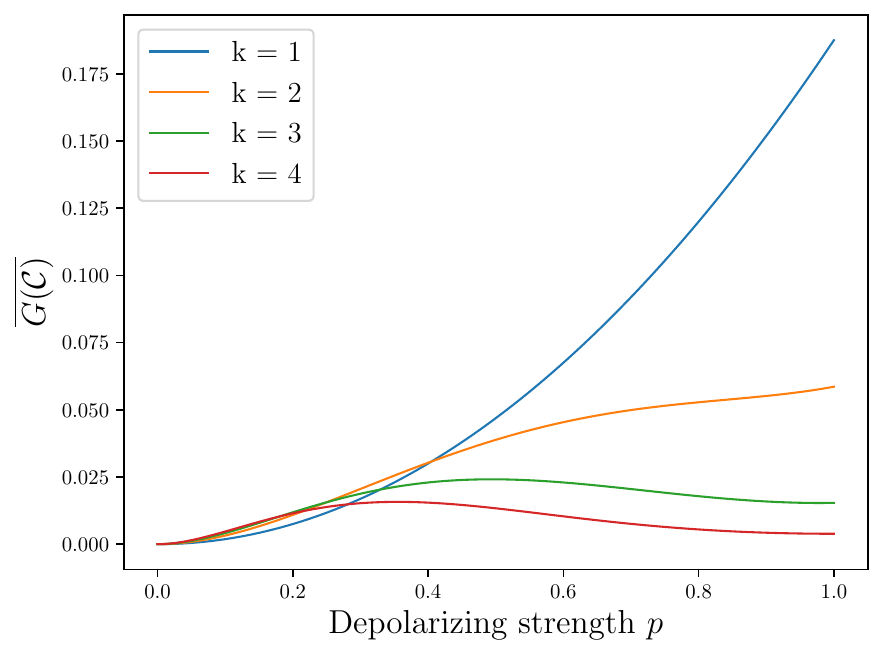}
    \caption{$\overline{G(\mathcal{C})}$ vs. $p$ for small values of $k$. As $k$ increases, $\overline{G(\mathcal{C})}$ flattens and eventually goes to $0$.}
    \label{depolarizing_k_plot}
\end{figure}

Again, we see the second corollary follows clearly from this formula. If $p > 0$, both terms have arguments less than $1$ in the exponential, and both go to $0$ as $k \rightarrow \infty$. This can be understood in terms of the competition between decoherence and scrambling. The vanishing of the first term implies maximal decoherence in terms of the bipartite degrees of freedom, leading to a vanishing bipartite $\mathcal{A}$-OTOC. All the information that would otherwise be scrambled across the bipartition has now leaked out of the system.

With depolarization, there are no values of $k$ for which we can achieve maximal scrambling. However, for any given finite $k$, $\overline{G(\mathcal{C})}$ is finite for all $p > 0$. Finite depolarization thus still produces macroscopic scrambling in infinite qubit chains. Moreover, we see that the behavior of $\overline{G(\mathcal{C})}$ with respect to $p$ varies significantly with $k$. For $k = 1$ and $k = 2$, $\overline{G(\mathcal{C})}$ varies monotonically with $p$. For $k \geq 3$, $\overline{G(\mathcal{C})}$ is non-monotonic in $p$, with a peak closer to $p=0$ as $k$ increases. Evidently, for all $k \geq 3$, there exists a critical noise strength $p^*(k)$ above which adding more noise increases decoherence more than it increases scrambling. The parameters $k$ and $p$ both control the strength of decoherence and of the perturbation that leads to scrambling. Since these two effects contribute oppositely in the bipartite $\mathcal{A}$-OTOC, one can generally expect the existence of an optimal combination of scrambling and decoherence. Consequently, as $k$ increases, the maximum occurs at a smaller $p^*(k)$, namely weaker perturbation and decoherence per qubit. The observation that, specifically for $k<3$, even maximum depolarization is insufficient to produce decoherence that surpasses scrambling lacks an obvious explanation. 

\subsection{Numerics}\label{numerics}

It is natural to wonder how general the phenomena outlined in this paper are. Namely, do unitaries composed of local random ``gates" display the butterfly effect? Can the Haar-randomness requirement be relaxed, and if so, what gate features correlate with the effect? We attempt to address these questions with numerical simulations. To this end, we replace the Haar-random $U$ with brickwork unitaries composed of $2L$ alternating layers of consecutive two-qubit gates, whose type we vary by circuit. We set $L = 8$ qubits.

We simulate two types of brickwork circuits. They are \textbf{1)} a circuit composed of Haar-random gates \textbf{2)} a circuit composed of the same, repeated dual-unitary gate \cite{dualunitary,EP}. Brickwork Haar-random circuits are known to be approximate k-designs at sufficient depth \cite{BrickworkTDesign}. With this knowledge and the fact that this model's bipartite $\mathcal{A}$-OTOC average requires a 4-design (see \cref{4designequation}), we know our simulated circuits are quite far from such a threshold, allowing us to observe a circuit that does not a priori approximate a Haar-random one. Repeated dual-unitary brickwork circuits are useful since their entangling properties are well understood. Two-qubit dual-unitary gates are parametrized by a value that corresponds simply to their operator space entangling power \cite{EP}. Operator space entangling power is a metric of the average operator entanglement \cite{zanardiEntanglementQuantumEvolutions2001} generated by a unitary acting on initially unentangled, random product operators with respect to a fixed bipartition. For two-qubit gates, we take the fixed bipartition to be the ``natural" one between the two qubits. By varying this parameter, we simulate dual-unitary circuits composed of varying operator space entangling power gates. Since scrambling is intimately connected with operator entanglement, studying the closeness of the dual-unitary circuits to the Haar-random result can shed insight into the extent to which operator entanglement generation is responsible for the butterfly effect in this model.


We calculate the bipartite $\mathcal{A}$-OTOC using \cref{aotoc_state}, implementing simulations via IBM's Qiskit package. To generate plots like the ones above, we sample 10 parameter ($\theta, \ p$, respectively) values for each example. For each value, we generate 5 random brickwork circuits. Finally, for each circuit, we find $G(\mathcal{C})$ for 5 different initial $\ket{\psi}$. From these, we calculate $\overline{G(\mathcal{C})}$ for each parameter value and generate plots. Each plot also contains the $L \rightarrow \infty$ analytic curve and the $L = 8$ curve.


For unitary noise, we apply a rotation about the $Z$ axis on a single qubit. Since our brickwork unitaries have no notion of qubit ``orientation," our results hold without loss of generality for any single-qubit unitary. We vary the rotation angle between $0$ and $\frac{\pi}{2}$. For depolarizing noise, we vary $p$ between $0$ and $1$ and apply the corresponding depolarization channel to a single qubit. 

\begin{figure}
    \centering
    \begin{subfigure}[t]{0.45\textwidth}
        \centering
        \adjustbox{valign=t}{\includegraphics[height=5cm]{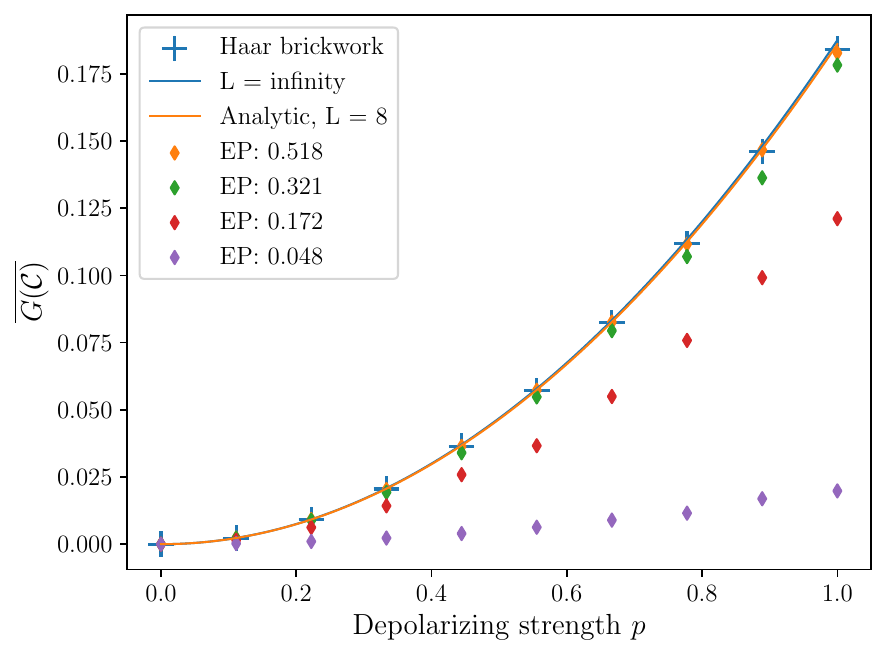}}
        
    \label{fig:depolarizing_noise_plot_2}
    \end{subfigure}
    \hfill
    \begin{subfigure}[t]{0.45\textwidth}
        \centering
        \adjustbox{valign=t}{\includegraphics[height=5cm]{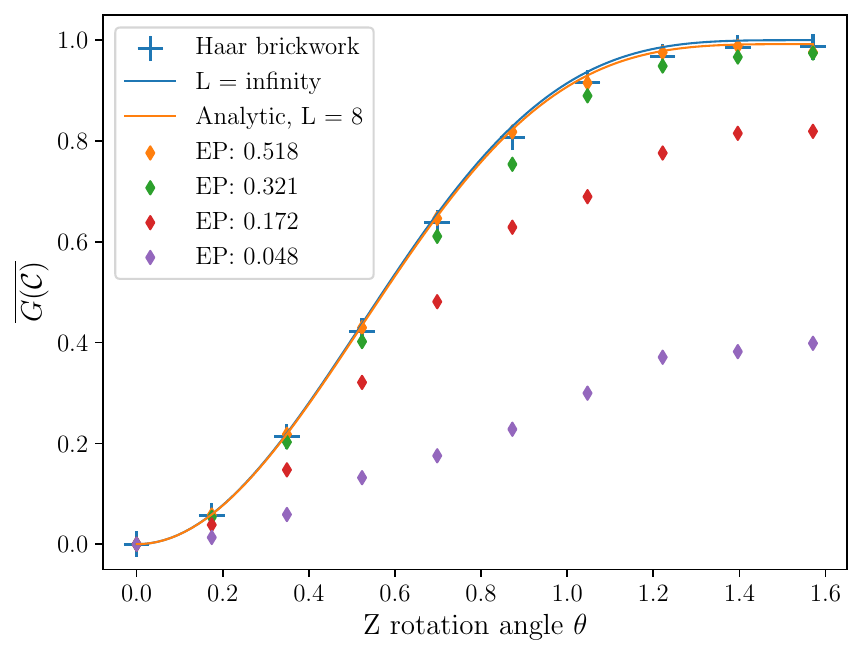}}
        \label{fig:unitary_noise_plot}
    \end{subfigure}

\caption{Depolarizing noise strength and unitary rotation angle vs. $\overline{G(\mathcal{C})}$. The ``EP" curves refer to repeated dual-unitary brickwork circuits with corresponding operator entangling power. The Haar-random circuits match the analytic results very closely. Increasing $E_p$ in dual-unitary brickwork circuits increases closeness to analytic results. The standard deviations for the EP curves under unitary noise are substantial enough for $\theta$ near $\pi/2$ that the non-monotonicity cannot be deemed significant. We omit the error bars for clarity.}
\end{figure}

The Haar-random match nearly identically with the analytic formulae. Interestingly, the high (but not maximal) operator space entangling power dual-unitary circuits also match nearly identically, indicating that the derived formulae may be approximated by a wider class of unitaries than just Haar-random. As the operator space entangling power of the dual-unitary gate decreases, so too does the bipartite $\mathcal{A}$-OTOC value. This correlates with a smaller delocalization of the perturbation from the initial qubit where it was applied, as the generation of operator entanglement decreases.

\begin{figure}
    \centering
    \begin{subfigure}[t]{0.45\textwidth}
        \centering
        \adjustbox{valign=t}{\includegraphics[height=5cm]{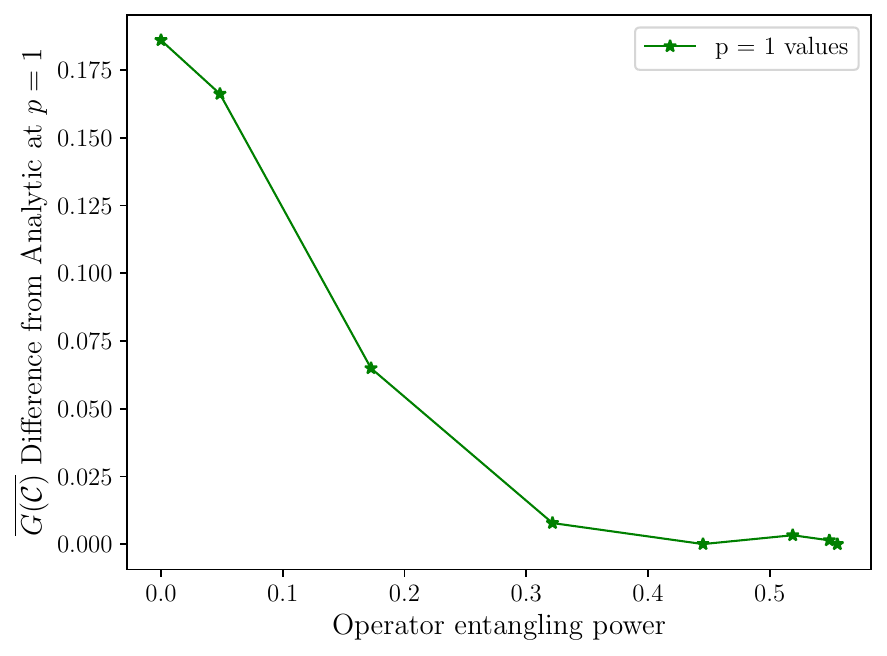}}
    \label{fig:depolarizing_deviation_plot}
    \end{subfigure}
    \hfill
    \begin{subfigure}[t]{0.45\textwidth}
        \centering
        \adjustbox{valign=t}{\includegraphics[height=5cm]{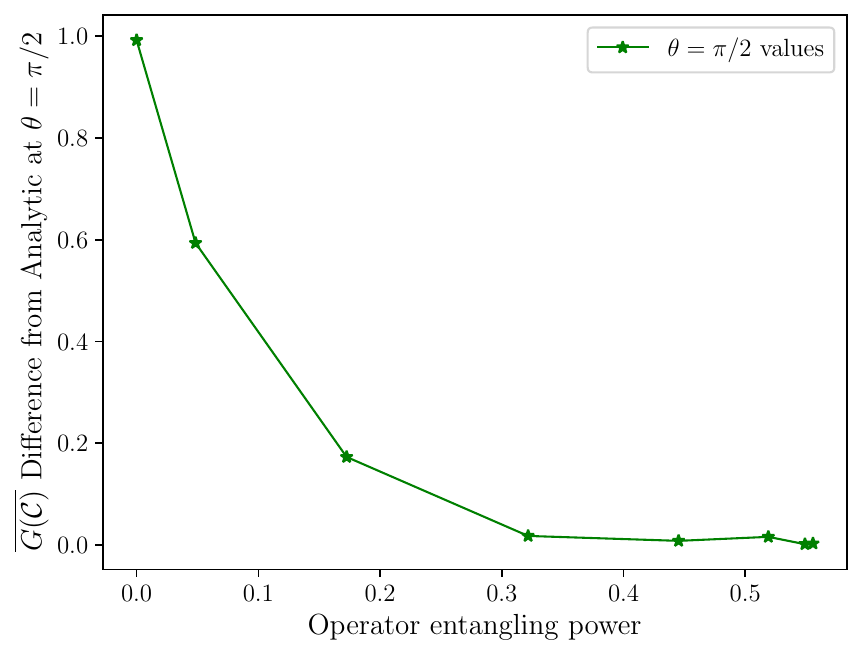}}
        \label{fig:unitary_deviation_plot}
    \end{subfigure}

\caption{Plots of deviation of dual-unitary brickworks from analytic results at argmax noise parameter values, for depolarizing and unitary noise respectively. Both curves seem to reflect a power law scaling of this difference with respect to $E_p$.}
\label{deviation_plot}
\end{figure}

We quantify this further in \cref{deviation_plot}. For both noise channels, we calculate $\overline{G(\mathcal{C})}-\overline{G(\mathcal{C}_{\text{DU}})}$ at the parameter values that maximize the analytic result ($p=1$ and $\theta = \pi/2$, respectively). This difference is plotted with respect to $E_p$. While a precise threshold for $E_p$ above which the typical behavior is well approximated may be dependent on the system size and the details of the circuit, \cref{deviation_plot} indicates that the operator space entangling power of the generating gate is one of the factors that contribute to the observation of the butterfly effect, even for classes of unitaries that do not form 4-designs.

\section{Numerics of Local Hamiltonian Circuits}

It is natural to question the extent to which this butterfly-like phenomenon generalizes beyond Haar-random unitaries. Haar-random unitaries require circuit depths extensive in qubit number to implement \cite{Schuster:2024ajb}, and as such do not correspond to readily preparable many-body dynamics for sufficiently large systems. Furthermore, an efficiently constructed, exact unitary 4-design remains elusive and approximate 4-designs themselves require extensively deep circuits \cite{Brandao:2012zoj}. It is thus conceivable that local many-body Hamiltonian dynamics fail to reproduce the effects in this paper, leaving our results merely abstract but unlikely to be observed.

To test this, we simulate noisy encoding-decoding circuits whose encodings are generated by a range of 1D spin chain Hamiltonians with periodic boundary conditions acting for different times $t$. We simulate XYZ Hamiltonians with transverse fields, given by
\begin{equation}
  \begin{aligned}
    H \;=\; -\sum_{j=1}^{L} \Bigl(&
        J_x\,\sigma_j^x \sigma_{j+1}^x
      + J_y\,\sigma_j^y \sigma_{j+1}^y
      + J_z\,\sigma_j^z \sigma_{j+1}^z \\[2pt]
      &+ h\,\sigma_j^z
    \Bigr),
  \end{aligned}
\end{equation}

as well as transverse-field Ising model (TFIM) Hamiltonians of the form

\begin{equation}
    H = -\sum_{j=1}^{L} \left( 
    J \sigma_j^x \sigma_{j+1}^x +
    h \sigma_j^z \right).
\end{equation}

The former is chaotic for general parameter choices \cite{xyz_chaotic_1, xyz_chaotic_2}, while the latter is always integrable via a mapping onto free fermions \cite{isingchainbeginners}. For each model, we simulate unitaries generated by $10$ different times $t$ under which $H$ is applied, interpolated between $0$ and $t_{max} \equiv \pi/\max(J_x, J_y, J_z, h)$, for both unitary and depolarizing noise on a single qubit in an $L = 8$ qubit circuit.

For the chaotic XYZ model presented here, we use $J_x = 1, J_y = 0.4, J_z = 0.6, h = 0.5$. For the TFIM, we use $J = 1, \ h = 0.8$.

\begin{figure}
    \centering
    \includegraphics[width=0.9\linewidth]{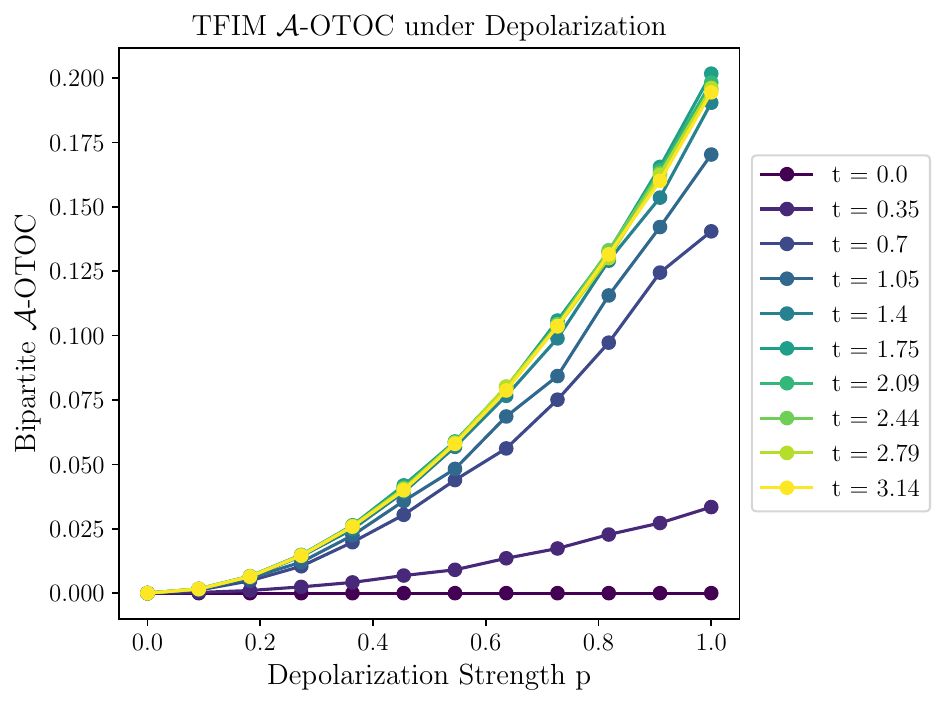}
    \caption{$G(\mathcal{C})$ vs. $p$ for $k = 1, \ L = 8$. Times $t$ range from $0$ to $\pi$. As $t$ increases, the curves approach the Haar-average result.}
    \label{tfim_manycurves}
\end{figure}

The plot in \cref{tfim_manycurves} shows that, even in a local integrable model, the bipartite $\mathcal{A}$-OTOC appears to approach Haar-random behavior as $t$ approaches a ``characteristic" timescale of the Hamiltonian.

\begin{figure}
    \centering
    \begin{subfigure}[t]{0.45\textwidth}
        \centering
        \adjustbox{valign=t}{\includegraphics[height=5cm]{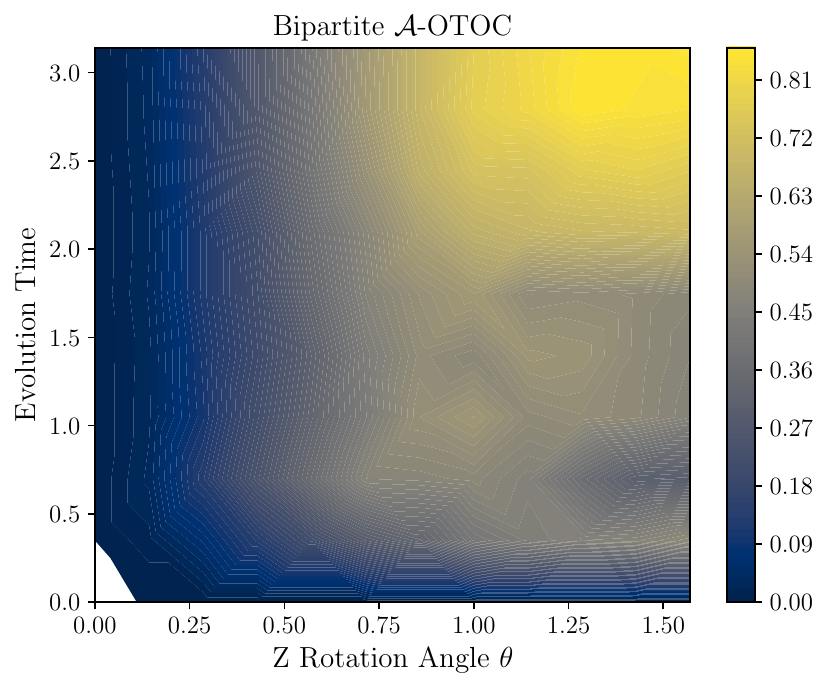}}
        
    \label{fig:tfim_ham_uni}
    \end{subfigure}
    \hfill
    \begin{subfigure}[t]{0.45\textwidth}
        \centering
        \adjustbox{valign=t}{\includegraphics[height=5cm]{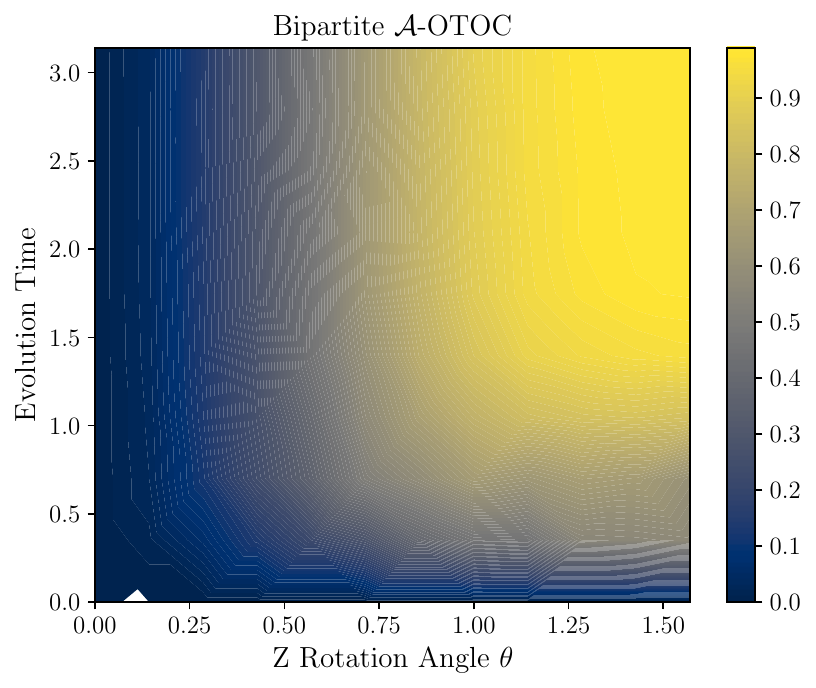}}
        \label{fig:heis_ham_uni}
    \end{subfigure}

\caption{Bipartite $\mathcal{A}$-OTOC vs. evolution time $t$ and unitary noise angle $\theta$ for \textbf{above)} TFIM Hamiltonian with $J = 1, \ h=0.8$, and \textbf{below)} chaotic XYZ model with $J_x = 1, J_y = 0.4, J_z = 0.6, h = 0.5$.}
\label{unizz}
\end{figure}

In \cref{unizz}, we see that the chaotic XYZ circuits display greater scrambling at earlier times and also achieve higher absolute values of the bipartite $\mathcal{A}$-OTOC than TFIM circuits. This is expected given that chaotic Hamiltonians tend to scramble information more rapidly than integrable ones \cite{otoc2}. However, both models give behavior at longer times that strongly resembles the Haar results. Similar behavior is displayed for depolarizing noise in \cref{depolzz}. These simulations indicate that the Haar-average behavior captures a butterfly effect that may hold for local Hamiltonian systems as well. 

 \begin{figure}[!ht]
    \centering
    \begin{subfigure}[t]{0.45\textwidth}
        \centering
        \adjustbox{valign=t}{\includegraphics[height=5cm]{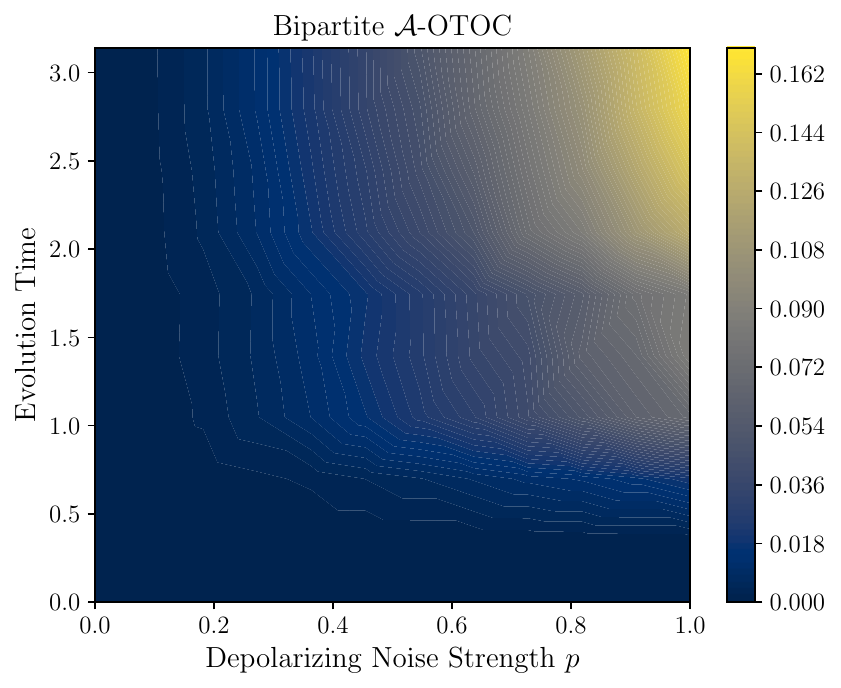}}
        
    \label{fig:tfim_ham_dep}
    \end{subfigure}
    \hfill
    \begin{subfigure}[t]{0.45\textwidth}
        \centering
        \adjustbox{valign=t}{\includegraphics[height=5cm]{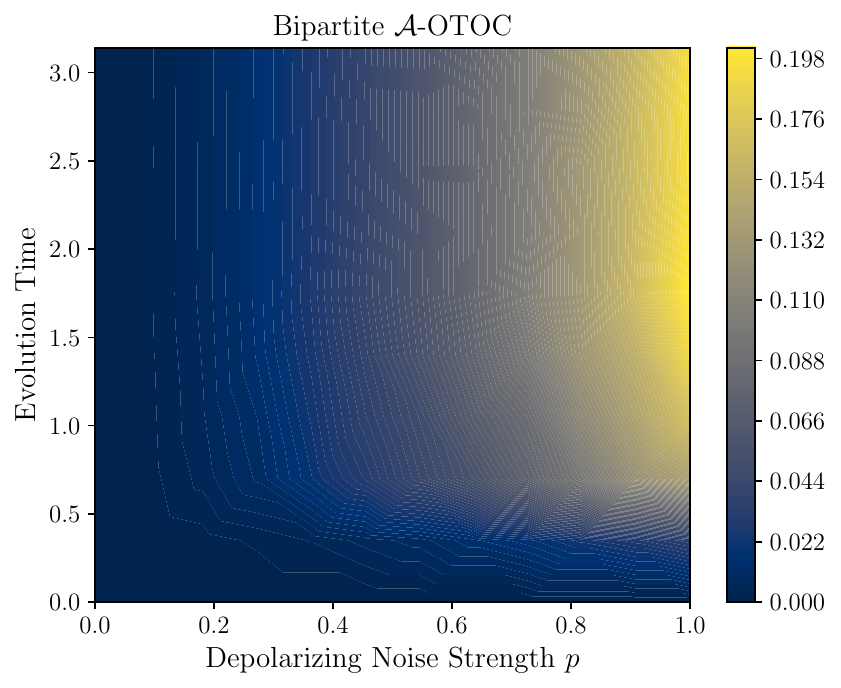}}
        \label{fig:heis_ham_dep}
    \end{subfigure}

\caption{Bipartite $\mathcal{A}$-OTOC vs. evolution time $t$ and depolarizing noise strength $p$ for \textbf{above)} TFIM Hamiltonian with $J = 1, \ h=0.8$, and \textbf{below)} chaotic XYZ model with $J_x = 1, J_y = 0.4, J_z = 0.6, h = 0.5$.}
\label{depolzz}
\end{figure}

\section{Conclusion}\label{conclusion}

We have studied a simple model circuit consisting of a Haar-random unitary $U$, identical noise $\mathcal{E}$ on $k$ qubits, and $U^\dagger$ to explore the interconnection between scrambling, noise, and chaos. We have derived an analytic formula for the bipartite $\mathcal{A}$-OTOC averaged over Haar-random $U$, which we employ as our scrambling metric. Remarkably, our formula drastically simplifies into an elegant expression when the system size $L \rightarrow \infty$.

From this formula, it immediately becomes clear that the bipartite $\mathcal{A}$-OTOC, which tracks scrambling between two macroscopic subsystems, remains finite even when noise is applied only to a finite subset of qubits. This is a rather surprising instance of a butterfly effect, in which noise on as little as one qubit triggers macroscopically observable scrambling in an infinite system. We also derive formulae for the bipartite $\mathcal{A}$-OTOC in circuits where $k \propto L$, highlighting a fundamental distinction between dissipative and unitary noise.

To better elucidate these results, we explicitly work through the examples of depolarizing and unitary ($R_z$) noise. In both cases, the bipartite $\mathcal{A}$-OTOC vs. noise strength curves display a strong $k$-dependence. For depolarizing noise, the functional shape of such curves varies sharply with $k$, providing another example of a butterfly-esque phenomenon -- the infinite system can somehow ``distinguish" between noise on one versus two qubits.

To help understand the generality of these phenomena, we have conducted numerical simulations on $8$-qubit circuits for both depolarizing and unitary noise, over multiple unitary types. Haar-random brickwork circuits fit very closely to the analytic Haar-random formulae. Brickwork circuits composed of highly operator space entangling dual-unitary gates closely fit the formulae as well, while those with weakly entangling dual-unitaries fall noticeably below. This preliminarily suggests that our results, and the existence of the observed butterfly effect, transcend the Haar-random unitary case. It also provides explicit evidence of the connection between gate entangling properties and circuit-wide scrambling.

The work here suggests several avenues for further theoretical development. First, the relationship between unitary type and the scrambling properties in this model is one that can be explored much further. Averages over Haar-random unitaries are convenient for their analytic tractability, but these are by no means representative of all physical processes. It would be interesting to investigate, with finite-size scaling methods, whether the numerical data we observe in our non-Haar random simulations reflect a butterfly effect in the limit $L \rightarrow \infty$. The connection between operator entangling power and the scaling of the bipartite $\mathcal{A}$-OTOC in this model is also worth exploring further, in the context of recent results on phase transitions in similar circuits \cite{Lovas,Turkeshi}; there may exist a similar threshold value above which there is a clear butterfly effect. Finally, it would be worthwhile to determine whether other information scrambling metrics display a similar effect and study their relationship to the behavior of the bipartite $\mathcal{A}$-OTOC.

This work also raises practical questions. The general relationship between the bipartite $\mathcal{A}$-OTOC and quantum channel capacity has yet to be explored. It may be interesting to explore it in this context -- the modulation of noise may be able to transmit information from one end of a spin chain to another. It is also conceivable that directly measuring the $\mathcal{A}$-OTOC could provide information about local noise in systems where global measurements are easier to perform than local ones. Finally, it would be an experimental achievement to directly measure this circuit model on quantum hardware.

\section{Acknowledgements}\label{acknowledgements}

ED and PZ acknowledge helpful discussions with Daniel Lidar. ED and PZ acknowledge partial support from the NSF award PHY2310227. FA acknowledges financial support from a University of Southern California "Philomela Myronis" scholarship. This research was (partially) sponsored by the Army Research Office and was accomplished under Grant Number W911NF-20-1-0075. The views and conclusions contained in this document are those of the authors and should not be interpreted as representing the official policies, either expressed or implied, of the Army Research Office or the U.S. Government. The U.S. Government is authorized to reproduce and distribute reprints for Government purposes notwithstanding any copyright notation herein.

\bibliographystyle{unsrt}
\bibliography{references}

\appendix\label{appendix}
\begin{widetext}

\section{Linearization of A-OTOC formula}

To get the $\mathcal{A}$-OTOC \cref{aotoc_def} into a form in which we can take a Haar average over $U$, we rely on successive applications of the equation

\begin{equation}\label{swap_identity}
    \Tr(S (A\otimes B)) = \Tr(AB),
\end{equation}

where $S$ is the swap operator.

The starting point of the derivation is the \cref{aotoc_swap}. Let $\mathcal{H}^{\otimes 4} \cong \otimes_i \mathcal{H}_i$, where each $\mathcal{H}_i \cong \mathcal{H}_{A_i} \otimes \mathcal{H}_{B_i}$. Henceforth in this section, we will take $S_{A_i A_j}$ to be the operator that swaps subsystem $A_i$ of Hilbert space copy $\mathcal{H}_i$ with subsystem $A_j$ of Hilbert space copy $\mathcal{H}_j$. We will take $S_{ij}$ to be the swap between $\mathcal{H}_i$ and $\mathcal{H}_j$.

To improve readability, let $d_B S_{ij}-S_{A_i A_j} \equiv M_{ij}$ (see \cref{aotoc_swap}), and let the Kraus operators of $\mathcal{E}^{\otimes 2}$ be $\{\Lambda_\alpha\}_\alpha$.

Rewriting \cref{aotoc_swap}, we have 
\begin{equation}
     G(\mathcal{C})=\frac{1}{d^2}\Tr(M_{12} (\mathcal{U}^\dagger)^{\otimes 2}\mathcal{E}^{\otimes 2}\mathcal{U}^{\otimes 2} (S_{A_1 A_2})).
\end{equation}

Expanding this gives
\begin{equation}
    G(\mathcal{C})=\frac{1}{d^2}\sum_\alpha \Tr((U^\dagger)^{\otimes 2} \Lambda_\alpha U^{\otimes 2} S_{A_1 A_2} (U^\dagger)^{\otimes 2} \Lambda_\alpha^\dagger U^{\otimes 2} M_{12}).
\end{equation}

Now we use the swap identity and expand from $\mathcal{H}^{\otimes 2}$ to $\mathcal{H}^{\otimes 4}$:
\begin{equation}
    = \frac{1}{d^2}\sum_\alpha \Tr(S_{13}S_{24}(U^\dagger)^{\otimes 4} (\Lambda_\alpha\otimes\Lambda_\alpha^\dagger) (U^{\otimes 4}) (S_{A_1 A_2} \otimes M_{34})).
\end{equation}

Finally, expanding out $M_{34}$, we get:
\begin{equation}\label{4designequation}
    G(\mathcal{C})=\frac{1}{d^2}\Tr\left(S_{(13)(24)} (U^\dagger)^{\otimes 4} \left(\sum_\alpha \Lambda_\alpha \otimes \Lambda_\alpha^\dagger \right) U^{\otimes 4} \left(S_{A_1 A_2}\otimes(d_B S_{(34)} - S_{A_3 A_4}) \right) \right).
\end{equation}

Now, to calculate $\overline{G(\mathcal{C})}$, we must to evaluate the following (and plug back into the above formula): $$\mathbb{E}_{U \in \mathcal{U}\left(2^L\right)}\left[\left(U^{\dagger}\right)^{\otimes 4} \left( \sum_\alpha \Lambda_\alpha \otimes \Lambda_\alpha^\dagger \right) U^{\otimes 4}\right].$$

Note that the set of Kraus operators $\{\Lambda_\alpha \}_\alpha$ of $\mathcal{E}^{\otimes 2}$ is just the set of all tensor products of Kraus operators of $\mathcal{E}$, namely $\{K_i\}_i$. Thus, 
\begin{equation}
    \sum_\alpha \Lambda_\alpha \otimes \Lambda_\alpha^\dagger = \sum_{ij} K_i \otimes K_j \otimes K_i^\dagger \otimes K_j^\dagger.
\end{equation}

This operator will take the place of $O$ in what follows for the final calculation.

\section{Schur-Weyl Duality}

Here, we follow the articulate notation of the supplementary material of \cite{Turkeshi}. To simplify the subsequent equations, let

\begin{equation}
    \Phi_{\text {Haar }}^{(k)}(O)=\mathbb{E}_{U \in \mathcal{U}\left(2^L\right)}\left[\left(U^{\dagger}\right)^{\otimes k} O U^{\otimes k}\right] \equiv \int_{\text {Haar }} d \mu(U)\left(U^{\dagger}\right)^{\otimes k} O U^{\otimes k}.
\end{equation}

Via Schur-Weyl duality \cite{weingarten_1}, it follows that:
\begin{equation}
    \Phi_{\text {Haar }}^{(k)}(O)=\sum_{\pi \in \mathcal{S}_k} b_\pi(O) T_\pi,
\end{equation}

where $\mathcal{S}_k$ is the permutation group over $k$ elements and $T_\pi$ is the representation of the permutation $\pi$ acting on $\left(\mathbb{C}^{2^L}\right)^{\otimes k}$. Note that $T_\pi=\otimes_{i=1}^N t_\pi^{(i)}$, where $t_\pi^{(i)}$ is the representation of $\pi$ on the $i$-th qubit, i.e. each of these permutations is the tensor of single-qubit permutations across the $k$ copies of the system Hilbert space.

It further follows that:
\begin{equation}\label{traces}
    b_\pi=\sum_{\sigma \in S_k} W_{\pi, \sigma} \Tr\left(O T_\sigma\right),
\end{equation}

where $W_{\pi, \sigma}$ is the \textit{Weingarten symbol} corresponding to the permutations $\pi, \sigma$. 

Recall that for the symmetric group $\mathcal{S}_k$, the irreps correspond to the integer partitions of $k$. An integer partition of $k$ is an (ordered, for uniqueness) set of positive integers which add up to $k$. Let $\lambda$ refer to an arbitrary integer partition of $k$, let $\Pi_\lambda$ be the projection onto the irrep (corresponding to) $\lambda$, $d_\lambda$ be the irrep's dimension, and $\chi_\lambda$ be its character function. Then
\begin{equation}
    W_{\pi, \sigma}=\sum_{\lambda \vdash k} \frac{d_\lambda^2}{(k!)^2} \frac{\chi_\lambda(\pi \sigma)}{\operatorname{tr}\left(\Pi_\lambda\right)}.
\end{equation}

To calculate the traces in \cref{traces}, we make use of the following identity. Let $T_\sigma \in \mathcal{S}_4$ have the cyclic representation $(ab)(cd)$, where $\left\{a...d\right\} \in \left\{1,2,3,4\right\}$\footnote{The choice of a $(2,2)$ cycle structure here is arbitrary and made to illustrate the identity.}. Let $O = \otimes_{k=1}^4 M_k$. Then

\begin{equation}
    \Tr(T_\sigma O) = \Tr(M_a M_b)\Tr(M_c M_d).
\end{equation}

This is just a simple generalization of the above swap formula.

Finally, to yield the formulae in \cref{mainresults}, we created a Mathematica notebook, making use of the excellent group-theoretic functions written in the notebook of \cite{Turkeshi}.

\section{Typicality of Averaged Quantities}\label{app_typicality}

In this section, we explicitly demonstrate a typicality result for $G(\mathcal{C})$ in the case of unitary noise $\mathcal{E}$. The proof does not straightforwardly extend to dissipative noise channels, though we offer some intuition suggesting similar results ought to hold.

We begin by defining:

\begin{definition}[Lipschitz Continuity]
    Let $f: U(d) \rightarrow \mathbb{R}$ such that
    \begin{equation}
        \lVert f(V) - f(W) \rVert \leq K \lVert V - W \rVert_2 \ \forall \ V, W \in U(d).
    \end{equation}

    Then $f$ is \textbf{Lipschitz continuous} with Lipschitz constant $K$.
\end{definition}

Now, we restate the well-known:

\begin{lemma}[Levy's Lemma for Unitary Operators\cite{CotlerLevyLemma}]

\[
\mathrm{Prob}\left\{ |f(U) - \mathbb{E}_U f(U)| \geq \epsilon \right\} \leq 2\exp\left( -\frac{2d\epsilon^2}{9\pi^3 K^2} \right).
\]

\end{lemma}

Now, we must solve for the Lipschitz constant of $G(\mathcal{C})$ with respect to $U$. To do this, we use the formula for the bipartite $\mathcal{A}$-OTOC under unitary noise $\mathcal{E} = \mathcal{W}$:

\begin{equation}
    G(\mathcal{C}) = 1-\frac{1}{d^2}\Tr\left(S_{(13)(24)} (\mathcal{U}^\dagger)^{\otimes 4} \left(\mathcal{W}^{\otimes 2} \right) \left(S_{(A_1A_2)(A_3 A_4)} \right) \right),
\end{equation}

where $\mathcal{U}^{\dagger}$ is the superoperator channel of $U^\dagger$.

We now proceed similarly in fashion to \cite{Zanardi2022quantumscramblingof}. We define the norm $\| \mathcal{T}\|_{2,2} \equiv \sup_{\|X\|_2 =1}\| \mathcal{T}(X)\|_2$. Let $\mathcal{C}_U=\mathcal{U}^\dagger \mathcal{W} \mathcal{U}$ and $\mathcal{C}_V=\mathcal{V}^\dagger \mathcal{W} \mathcal{V}$, where $\mathcal{U}, \mathcal{W}, \mathcal{V}$ are unitary channels.

We have:

\begin{equation}
    \lvert G(\mathcal{C}_U) - G(\mathcal{C}_V) \rvert = \left\lvert \frac{1}{d^2} \left\langle S_{(A_1 A_2)}, ({\mathcal{U}^\dagger}^{\otimes 2} \mathcal{W}^{\otimes 2} \mathcal{U}^{\otimes 2} - {\mathcal{V}^\dagger}^{\otimes 2} \mathcal{W}^{\otimes 2} \mathcal{V}^{\otimes 2})(S_{(A_1 A_2)}) \right\rangle \right\rvert.
\end{equation}

Using the Cauchy-Schwarz inequality, we see this expression is

\begin{equation}
    \leq \frac{1}{d^2} \lVert S_{(A_1 A_2)} \rVert_2^2 \; \lVert {\mathcal{U}^\dagger}^{\otimes 2} \mathcal{W}^{\otimes 2} \mathcal{U}^{\otimes 2} - {\mathcal{V}^\dagger}^{\otimes 2} \mathcal{W}^{\otimes 2} \mathcal{V}^{\otimes 2} \rVert_{2,2}.
\end{equation}

Since $S_{(A_1 A_2)}$ is unitary, its 2-norm squared equals $d^2$. Using this, and adding and subtracting $\mathcal{V}^\dagger \mathcal{W} \mathcal{V} \otimes \mathcal{U}^\dagger \mathcal{W}\mathcal{U}$, we get:

\begin{equation}
    = \lVert (\mathcal{U}^\dagger \mathcal{W} \mathcal{U} - \mathcal{V}^\dagger \mathcal{W} \mathcal{V}) \otimes \mathcal{U}^\dagger \mathcal{W} \mathcal{U} + \mathcal{V}^\dagger \mathcal{W} \mathcal{V} \otimes (\mathcal{U}^\dagger \mathcal{W} \mathcal{U} - \mathcal{V}^\dagger \mathcal{W} \mathcal{V})\rVert_{2,2}.
\end{equation}

Using the triangle inequality, we can upper bound this by:

\begin{equation}
    \leq 2 \lVert \mathcal{U}^\dagger \mathcal{W} \mathcal{U} - \mathcal{V}^\dagger \mathcal{W} \mathcal{V} \rVert_{2,2} = 2 \sup_{\lVert X \rVert_2 =1} \lVert (\mathcal{U}^\dagger \mathcal{W} \mathcal{U} - \mathcal{V}^\dagger \mathcal{W} \mathcal{V})(X) \rVert_2.
\end{equation}

Now, we write out the superoperator explicitly in terms of its unitaries:

\begin{equation}
    = 2 \sup_{\lVert X \rVert_2 =1} \lVert (U^\dagger W U - V^\dagger W V) X (U^\dagger W^\dagger U) + (V^\dagger W V) X (U^\dagger W^\dagger U - V^\dagger W^\dagger V) \rVert_2 .
\end{equation}

We again apply the triangle inequality to get:

\begin{equation}
    \leq 2 \sup_{\lVert X \rVert_2 =1} \lVert \mathcal{U}^\dagger \mathcal{W} \mathcal{U} - \mathcal{V}^\dagger \mathcal{W} \mathcal{V} \rVert_2 \, \lVert X U^\dagger W^\dagger U \rVert_2 + \lVert V^\dagger W V X \rVert_2 \, \lVert U^\dagger W^\dagger U - V^\dagger W^\dagger V \rVert_2.
\end{equation}

Note that each norm term containing $X$ is the norm of $X$ multiplied by a unitary, which preserves the norm. Thus, both of these terms equal $\lVert X \rVert_2 = 1$. This simplies the above expression, and allows us to repeat the ``telescoping" trick of adding and subtracting $V^\dagger W U$ (though these are now operators instead of superoperators). We arrive at:

\begin{equation}
    = 4 \lVert U^\dagger W U - V^\dagger W V \rVert_2 = 4 \lVert (U^\dagger - V^\dagger) W U + V^\dagger W (U-V) \rVert_2.
\end{equation}

Finally, we apply the triangle inequality again to arrive at:

\begin{equation}
    \leq 4 (\lVert (U^\dagger - V^\dagger)W U \rVert_2 + \lVert V^\dagger W (U-V) \rVert_2 ) = 8 \lVert U - V \rVert_2.
\end{equation}

We can then plug the Lipschitz constant $K = 8$ into Levy's lemma, yielding:

\begin{equation}
    \mathrm{Prob}\left\{ |G(\mathcal{C}_U) - \overline{G(\mathcal{C})} \geq \epsilon \right\} \leq 2\exp\left( -\frac{d\epsilon^2}{288\pi^3} \right).
\end{equation}

As $L \rightarrow \infty$, fluctuations about the mean are exponentially suppressed by $d \rightarrow \infty$. The Haar-average result is thus typical as $L \rightarrow \infty$.

This proof does not immediately flow through to the case of dissipative noise due to the pernicious factor of $d_B$ in \cref{aotoc_swap} that prevents the derivation of a Lipschitz constant independent of system size. However, we have some intuition that similar behavior should be displayed in the thermodynamic limit. First, we numerically observe low variance about the mean in our simulations of noisy Haar-random brickwork circuits. Furthermore, the value of $\overline{G(\mathcal{C})}$ for dissipative noise has a strictly smaller upper bound than for coherent noise; it seems that these lowers values should also curtail the possibility of large fluctuations. Finally, along similar lines, dissipative noise leaks information outside the system. One may expect that this leakage also leaks some information about the unitary encoding-decoding channel conjugation the noise, thereby reducing the differences in behavior among different unitaries. Of course, these do not constitute a proof, and a formal derivation is left for future work.

\end{widetext}

\end{document}